\documentclass{llncs}
\usepackage[T1]{fontenc}
\usepackage{graphicx, algorithm, algorithmic}
\usepackage[utf8]{inputenc}
\usepackage{tikz}
\usetikzlibrary{shapes.geometric,shapes.misc}
\usetikzlibrary{calc,shapes,arrows,arrows.meta,shapes.misc,decorations.pathreplacing}
\usetikzlibrary{positioning,calc}
\usetikzlibrary{automata, positioning}

\usepackage{tikz-qtree}
\usepackage{subcaption}
\usepackage{lipsum}
\usepackage{listings}
\usepackage{adjustbox}
\usepackage{amsmath,amssymb}
\usepackage{hyperref}
\usepackage{ textcomp }
\usepackage{ wasysym }
\usepackage{bm}

\def\balpha{{\boldsymbol{\alpha}}}

\def\blim#1{\ensuremath{{\boldsymbol{Lim}(#1)}}}

\let\doendproof\endproof
\renewcommand\endproof{~\hfill~$\Box$\doendproof}

\usepackage{xcolor}

\definecolor{dkgreen}{rgb}{0,0.6,0}
\definecolor{gray}{rgb}{0.5,0.5,0.5}
\definecolor{mauve}{rgb}{0.58,0,0.82}

\lstdefinestyle{myScalastyle}{
	frame=tb,
	language=scala,
	aboveskip=3mm,
	belowskip=3mm,
	showstringspaces=false,
	columns=flexible,
	basicstyle={\small\ttfamily},
	numbers=none,
	numberstyle=\tiny\color{gray},
	keywordstyle=\color{blue},
	commentstyle=\color{dkgreen},
	stringstyle=\color{mauve},
	frame=single,
	breaklines=true,
	breakatwhitespace=true,
	tabsize=3,
	numbers=left,
	deletekeywords={for},
	otherkeywords={function,==,=, Let}
}

\begin{document}
	\title{The order type of scattered context-free orderings of rank one is computable}
	\author{Kitti~Gelle,~Szabolcs~Iv\'an}
	\institute{Department of Computer Science, University of Szeged, Hungary\\ \email{\{kgelle,szabivan\}@inf.u-szeged.hu}}
	
	\maketitle
	\begin{abstract}
	A linear ordering is called context-free if it is the lexicographic ordering of some context-free language
	and is called scattered if it has no dense subordering. Each scattered ordering has an associated ordinal,
	called its rank. It is known that the isomorphism problem of scattered context-free orderings is
	undecidable, if one of them has a rank at least two. In this paper we show that it is decidable
	whether a context-free ordering has rank at most one, and if so, its order type is effectively computable.
	\end{abstract}

\section{Introduction}
If an alphabet $\Sigma$ is equipped by a linear order $<$, this order can be extended to the lexicographic
ordering $<_\ell$ on $\Sigma^*$ as $u<_\ell v$ if and only if either $u$ is a proper prefix of $v$ or
$u=xay$ and $v=xbz$ for some $x,y,z\in\Sigma^*$ and letters $a<b$. So any language $L\subseteq \Sigma^*$
can be viewed as a linear ordering $(L,<_\ell)$. Since $\{a,b\}^*$ contains the dense ordering
$(aa+bb)^*ab$ and every countable linear ordering can be embedded into any countably infinite dense ordering,
every countable linear ordering is isomorphic to one of the form $(L,<_\ell)$ for some language $L\subseteq\{a,b\}^*$.
A linear ordering (or an order type) is called \emph{regular} or \emph{context-free}
if it is isomorphic to the linear ordering (or, is the order type) of some language of the appropriate type.
It is known~\cite{DBLP:journals/fuin/BloomE10} that an ordinal is regular if and only if it is less than $\omega^\omega$
and is context-free if and only if it is less than $\omega^{\omega^\omega}$. Also, the Hausdorff rank~\cite{rosenstein}
of any scattered regular (context-free, resp.) ordering is less than $\omega$ ($\omega^\omega$, resp)~\cite{ITA_1980__14_2_131_0,10.1007/978-3-642-29344-3_25}.

It is known~\cite{GelleIvanTCS} that the order type of a well-ordered language generated by a prefix grammar (i.e. in which
each nonterminal generates a prefix-free language)
is computable,
thus the isomorphism problem of context-free ordinals is decidable if the ordinals in question are given as the lexicograpic
ordering of \emph{prefix} grammars.
Also, the isomorphism problem of regular orderings is decidable as well~\cite{DBLP:journals/ita/Thomas86,BLOOM200555}.
At the other hand, it is undecidable for a context-free grammar whether it generates a dense language,
hence the isomorphism problem of context-free orderings in general is undecidable~\cite{ESIK2011107}.

Algorithms that work for the well-ordered case can in many cases be ``tweaked'' somehow
to make them work for the scattered case as well:
e.g. it is decidable whether $(L,<_\ell)$ is well-ordered or scattered~\cite{10.1007/978-3-642-22321-1_19}
and the two algorithms are quite similar.
In an earlier paper~\cite{GelleIvanGandalf} we showed that it is undecidable for a scattered context-free ordering
of rank $2$
whether its order type is $\omega+(\omega+\zeta)\times\omega$, even if it is given by a prefix grammar -- so
the complexity of the isomorphism problem
is quite different when one makes the step from well-ordered languages to scattered ones.

In the current paper we complement this result by showing that if the rank of a scattered context-free
ordering is at most one (we also show that this property is decidable as well), then its order type is
effectively computable (as a finite sum of the order types $1$, $\omega$ and $-\omega$).
	
\section{Notation}
A \emph{linear ordering} is a pair $(Q,<)$, where $Q$ is some set and the $<$ is a transitive, antisymmetric and connex 
(that is, for each $x,y\in Q$ exactly one of $x<y$, $y<x$ or $x=y$ holds) binary relation on $Q$.
The pair $(Q,<)$ is also written simply $Q$ if the ordering is clear from the context.
A (necessarily injective) function $h: Q_1\to Q_2$, where $(Q_1,<_1)$ and $(Q_2,<_2)$ are some linear orderings,
is called an \emph{(order) embedding}
if for each $x,y\in Q_1$, $x<_1 y$ implies $h(x) <_2 h(y)$. If $Q_1$ can be embedded into $Q_2$, then
this is denoted by $Q_1\preceq Q_2$.
If $h$ is also surjective, $h$ is an \emph{isomorphism}, in which case the two orderings are \emph{isomorphic}.
An isomorphism class is called an \emph{order type}. The order type of the linear ordering $Q$ is denoted by $o(Q)$.

For example, the class of all linear orderings contain all the finite linear orderings
and the orderings of the integers ($\mathbb{Z}$), the positive integers ($\mathbb{N}$) and the negative integers ($\mathbb{N}_{-}$) whose order type is denoted $\zeta$, $\omega$ and $-\omega$ respectively.
Order types of the finite sets are denoted by their cardinality, and $[n]$ denotes $\{1,\ldots,n\}$ for each $n\geq 0$, ordered in the standard way.

The ordered sum $\sum_{x\in Q} Q_x$, where $Q$ is some linear ordering and for each $x \in Q$, $Q_x$ is a linear ordering,
is defined as the ordering with domain $\{(x,q):x\in Q,q\in Q_x\}$ and ordering relation
$(x,q)<(y,p)$ if and only if either $x<y$, or $x=y$ and $q<p$ in the respective $Q_x$. 
If each $Q_x$ has the same order type $o_1$ and $Q$ has order type $o_2$, then the above sum has order type $o_1\times o_2$.
If $Q=[2]$, then the sum is usally written as $Q_1+Q_2$.

If $(Q,<)$ is a linear ordering and $Q'\subseteq Q$, we also write $(Q',<)$ for the subordering of $(Q,<)$, that is, to ease notation we also use $<$ for the restriction of $<$ to $Q'$.

A linear ordering $(Q,<)$ is called \emph{dense} if it has at least two elements and for each $x,y\in Q$ where $x<y$ there exists a $z\in Q$ such that $x<z<y$.
A linear ordering is \emph{scattered} if no dense ordering can be embedded into it.
It is well-known that every scattered sum of scattered linear orderings is scattered, and any finite union of scattered linear orderings is scattered.
A linear ordering is called a \emph{well-ordering} if it has no subordering of type $-\omega$. Clearly, any well-ordering is scattered. Since isomorphism preserves well-orderedness
or scatteredness, we can call an order type well-ordered or scattered as well,
or say that an order type embeds into another. We also write $o_1\preceq o_2$ to denote $o_1$ embeds into $o_2$.
The well-ordered order types are called \emph{ordinals}.
For any set $\Omega$ of ordinals, $(\Omega,<)$ is well-ordered by the relation $o_1\prec o_2~\Leftrightarrow~\hbox{``}o_1\hbox{ can be embedded injectively into }o_2$ but not vice versa''.
Amongst ordinals, it is common to use the notation $o_1<o_2$ instead of $o_1\prec o_2$.
The principle of well-founded induction can be formulated as follows. Assume $P$ is a property of ordinals such that for any ordinal $o$, if $P$ holds for all ordinals
smaller than $o$, then $P$ holds for $o$. Then $P$ holds for all the ordinals.

For standard notions and useful facts about linear orderings see e.g.~\cite{rosenstein} or~\cite{jalex}.

Hausdorff classified the countable scattered linear orderings with respect to their rank.
We will use the definition of the Hausdorff rank from \cite{10.1007/978-3-642-29344-3_25},
which slightly differs from the original one (in which $H_0$ contains only the empty ordering and the singletons,
and the classes $H_\alpha$ are not required to be closed under finite sum, see e.g.~\cite{rosenstein}).
For each countable ordinal $\alpha$, we define the class $H_\alpha$ of countable linear orderings as follows.
$H_0$ consists of all finite linear orderings, and
when $\alpha> 0$ is a countable ordinal,
then $H_\alpha$ is the least class of linear orderings closed under finite ordered sum and isomorphism
which contains all linear orderings of the form $\sum_{i\in\mathbb{Z}} Q_i$,
where each $Q_i$ is in $H_{\beta_i}$ for some $\beta_i < \alpha$.

By Hausdorff's theorem, a countable linear order $Q$ is scattered if and only if it belongs to $H_\alpha$ for some countable ordinal $\alpha$.
The \emph{rank} $r(Q)$ of a countable scattered linear ordering is the least ordinal $\alpha$ with $Q \in H_\alpha$.

As an example, $\omega$, $\zeta$, $-\omega$ and $\omega+\zeta$ or any finite sum of the form $\mathop\sum\limits_{i\in[n]}o_i$ with $o_i\in\{\omega, -\omega,1\}$ for each $i\in[n]$
each have rank $1$ while $(\omega+\zeta)\times\omega$ has rank $2$.

Let $\Sigma$ be an alphabet (a finite nonempty set) and let $\Sigma^*$ ($\Sigma^+$, resp) stand for the set of all (all nonempty, resp)
finite words over $\Sigma$, $\varepsilon$ for the empty word, $|u|$ for the length of the word $u$, $u\cdot v$ or simply $uv$ for the concatenation of $u$ and $v$. A \emph{language} is an arbitrary subset $L$ of $\Sigma^*$. 
We assume that each alphabet is equipped by some (total) linear order.
Two (strict) partial orderings, the strict ordering $<_s$ and the prefix ordering $<_p$ are defined over $\Sigma^*$ as
follows:
\begin{itemize}
	\item $u <_s v$ if and only if $u = u_1au_2$ and $v = u_1bv_2$ for some words $u_1, u_2, v_2\in\Sigma^*$ and letters $a < b$,
	\item $u <_p v$ if and only if $v = uw$ for some nonempty word $w\in\Sigma^*$.
\end{itemize} 
The union of these partial orderings is the lexicographical ordering $<_\ell = <_s \cup <_p$.
We call the language $L$ well-ordered or scattered,
if $(L,<_\ell)$ has the appropriate property and we define the rank $r(L)$ of a scattered language $L$
as $r(L,<_\ell)$. The order type $o(L)$ of a language $L$ is the order type of $(L,<_\ell)$.
For example, if $a<b$, then $o\Bigl(\{a^kb:k\geq 0\}\Bigr)=-\omega$ and $o\Bigl(\{(bb)^ka:k\geq 0\}\Bigr)=\omega$.

When $\varrho$ is a relation over words (like $<_\ell$ or $<_s$), we write $K\varrho L$ if $u\varrho v$ for each word $u\in K$ and $v\in L$.

An \emph{$\omega$-word} over $\Sigma$ is an $\omega$-sequence $a_1a_2\ldots$ of letters $a_i\in\Sigma$. The set of all $\omega$-words over $\Sigma$ is denoted $\Sigma^\omega$.
The orderings $<_\ell$ and $<_p$ are extended to $\omega$-words. An $\omega$-word $w$ is called \emph{regular} if $w=uv^\omega=uvvvv\ldots$ for some finite words
$u\in\Sigma^*$ and $v\in\Sigma^+$. When $w$ is a (finite or $\omega$-) word over $\Sigma$ and $L\subseteq\Sigma^*$ is a language, then $L_{<w}$
stands for the language $\{u\in L:u<w\}$. Notions like $L_{\geq w}$, $L_{<_sw}$ are also used as well, with the analogous semantics.

A \emph{context-free grammar} is a tuple $G=(N,\Sigma,P,S)$, where $N$ is the alphabet of the \emph{nonterminal symbols}, $\Sigma$ is the alphabet of \emph{terminal symbols} (or \emph{letters}) which is disjoint from $N$, $S\in N$ is the \emph{start symbol} and $P$ is a finite set of \emph{productions} of the form $A\to\alpha$, where $A\in N$ and $\alpha$ is a \emph{sentential form}, that is, $\alpha = X_1X_2\ldots X_k$ for some $k\geq 0$ and $X_1,\ldots,X_k\in N\cup \Sigma$. The derivation relations $\Rightarrow$, $\Rightarrow_\ell$, $\Rightarrow^*$ and $\Rightarrow_\ell^*$ are defined as usual
(where the subscript $\ell$ stands for ``leftmost'').
The \emph{language generated by} a grammar $G$ is defined as $L(G) = \{u\in\Sigma^* ~|~ S\Rightarrow^*u\}$.
Languages generated by some context-free grammar are called \emph{context-free languages}. 
For any set $\Delta$ of sentential forms, the language generated by $\Delta$ is $L(\Delta) = \{u \in \Sigma^* ~|~ \alpha\Rightarrow^* u\hbox{ for some } \alpha\in\Delta\}$. 
As a shorthand, we define $o(\Delta)$ as $o(L(\Delta))$.
When $X,Y\in N\cup\Sigma$ are symbols of a grammar $G$, we write $Y\preceq X$ if $X\Rightarrow^*uYv$ for some words $u$ and $v$; $X\approx Y$ if $X\preceq Y$ and $Y\preceq X$ both hold;
and $Y\prec X$ if $Y\preceq X$ but not $X\preceq Y$. A production of the form $X\to X_1\ldots X_n$ with $X_i\prec X$
for each $i\in[n]$ is called an \emph{escaping production}.

A \emph{regular language} over $\Sigma$ is one which can be built up from the singleton languages $\{a\}$, $a\in\Sigma$
and the empty language $\emptyset$ with finitely many applications of taking (finite) union,
concatenation $KL=\{uv:u\in K,v\in L\}$ and iteration $K^*=\{u_1\ldots u_n:n\geq 0,u_i\in K\}$. 
For standard notions on regular and context-free languages the reader is referred to any standard textbook, such as~\cite{Hopcroft+Ullman/79/Introduction}.

Linear orderings which are isomorphic to the lexicographic ordering of some context-free (regular, resp.) language
are called \emph{context-free (regular, resp.) orderings}.	
\section{Limits of languages}
In this section we introduce the notion of a limit of a language and establish a connection:
the main contribution of this concept is that one can decide whether a context-free language has a finite
number of limits and if so, one can effectively compute the limits themselves (Lemma~\ref{lem-limx-finite}),
and that a language has a finite number of limits if and only if its order type is scattered of rank at most one
(Theorem~\ref{thm-finite-limits-imply-computability}).

Firstly, we recall (and prove for the sake of completeness) that $\Sigma^\infty$ forms a complete lattice
with the partial ordering $\leq_\ell$ (which can be turned into a metric space as well).
\begin{lemma}
	\label{lem-lattice}
	$(\Sigma^\infty, \leq_\ell)$ is a complete lattice.
\end{lemma}
\begin{proof}
	Let $L$ be an infinite language. If $L$ has a maximal element, then it is the supremum, otherwise we generate the word $a_1a_2a_3\ldots\in\Sigma^\omega$ in the following way: let $u_0 = \varepsilon$ and $u_i = a_1\ldots a_i$.
	We choose the largest possible letter $a_{i+1}$ with $(u_ia_{i+1})^{-1}L$ being nonempty.
	The word generated by this way is the supremum of $L$. 
\end{proof}

\subsection{Limits in general}
Though limits of a language could be defined as limits of Cauchy sequences in the aforementioned metric space,
the following (equivalent) definition is more convenient for our purposes.
\begin{definition}
	The word $w\in\Sigma^\omega$ is a \emph{limit} of an infinite language $L$, if for each $w_0<_p w$ there exists a word $u\in L$ such that $w_0<_p u$.
\end{definition}

If $L \subseteq \Sigma^*$ is a language, then we denote the limits of $L$ with $\blim{L}$. 

\begin{lemma}
	If $w$ is a limit of an infinite language $L$, then for each $w_0<_p w$ there exist infinitely many words $u\in L$ with $w_0<_p u$.
\end{lemma}
\begin{proof}
	We construct two sequences $w_0 <_p w_1 <_p \ldots <_p w$ and $u_0,\ u_1, \ldots \in L$ such that $w_i <_p u_i$ for each $i$ with mutual induction.
	Now $w_0$ is given. By definition for each $i$ there exists an $u_i\in L$ such that $w_i<_p u_i$ and $w_i\in \mathbf{Pref}(w)$ is constructed such that $|u_{i-1}|+1<|w_i|$.
	
	It is clear the words $u_i$ are pairwise different, since each $w_i$ has different length. So we get that $w_0$ is a prefix of each $u_i$, so $w_0$ is a prefix of infinitely many words in $L$.	
\end{proof}

\begin{lemma}	
\label{lem-infinite}
	For each infinite language $L$, the set $\blim{L}$ is nonempty.
\end{lemma}
\begin{proof}
	We construct a limit word $w=a_1a_2\ldots\in\Sigma^\omega$. Let $u_0 = \varepsilon$ and $u_i = a_1\ldots a_i$ and we choose $a_{i+1}\in\Sigma$ such that $u_ia_{i+1}$ is a prefix of infinitely many words in $L$. Since $L$ is infinite there exists such a letter. Thus we can construct an infinite word which is a limit of $L$.    
\end{proof}

Now we justify using the name ``limit'': any supremum or infimum of a chain (which is a Cauchy sequence) of words
of a language is a limit of the language.
\begin{lemma} 
	\label{lem-sup-is-limit}
	If $w_0, w_1, \ldots$ is a $<_\ell$ (or $>_\ell$ respectively) chain in $L$, then its supremum (infimum, resp.) is a limit of $L$.
\end{lemma}
\begin{proof}
	Let $w$ be the supremum of the $<_\ell$ chain in $L$. We only have to show that for each $w'<_p w$ there exists a member $u$ of the chain with $w'<_p u$.
	First, as the chain is infinite, there are infinitely many words $w_i$ with $|w_i|>|w'|$ and the supremum of these words $w_i$ is $w$ as well.
	For these words we cannot have $w'<_pw_i$.
	As $w>w'$ is the supremum of the words $w_i$, we have $w_i<_\ell w$ for each of them, moreover, there exists some word $w_i$ with $w'<_\ell w_i<_\ell w$.
	Since we know that $w'<_pw$, it has to be the case that $w'<_pw_i<_\ell w$ and the claim is proved.
	
	For the other case, let $w$ be the infimum of the $>_\ell$ chain in $L$ and let $w'<_pw$ be a prefix of $w$.
	We again have to show that $w'<_pw_i$ for some $i$. Again, we can take the subchain consisting of those words $w_i$ with $|w_i|>|w'|$,
	this does not change their infimum. Now we have $w'<_pw<_sw_i$ for each of these words $w_i$. We claim that $w'<_pw_i$ for at least one
	index $i$. Assume to the contrary that $w'<_sw_i$ for each $i$ and let us write $w'=u_0az^t$ where $z$ is the largest letter of $\Sigma$,
	and $a<z$.
	(Since there exist words with $w'<_sw_i$, $w'$ cannot have the form $z^t$). Let $b$ be the successor letter of $a$ and consider
	the word $w''=u_0b$. Obviously, $w''$ is the least word with $w'<_sw''$, thus $w''\leq_\ell w_i$ for each $i$. Since $w$ is the infimum of
	these words $w_i$, we should have $w''\leq_\ell w$ but this contradicts to $w'<_pw$ and $w'<_sw''$ as these two imply $w<_sw''$.
\end{proof}

Now we recall from~\cite{GelleIvanGandalf} that for any context-free language, we can compute a supremum or infimum
of the language.

\begin{lemma}[\cite{GelleIvanGandalf}, Lemma 1]
	\label{lem-sequence}
	For each sentential form $\alpha$ with $L(\alpha)$ being infinite, we can generate a sequence $w_0,w_1,\ldots \in L(\alpha)$ and a regular word $w\in\Sigma^\omega$ satisfying one of the following cases:
	\begin{itemize}
		\item[i)] $w_1<_sw_2<_s\ldots$ and $w=\mathop\bigvee\limits_{i\geq 0}w_i$
		\item[ii)] $w_1>_sw_2>_s\ldots$ and $w=\mathop\bigwedge\limits_{i\geq 0}w_i$
		\item[iii)] $w_1<_pw_2<_p\ldots$ and $w=\mathop\bigvee\limits_{i\geq 0}w_i$
	\end{itemize}
\end{lemma}
Hence, Lemma~\ref{lem-sup-is-limit} in conjunction with Lemma~\ref{lem-sequence}
ensure that whenever $L$ is an infinite context-free language,
then one of its limits can be effectively computed, and this particular limit will be a regular word.

Next, we show how to compute limits of unions and products:

\begin{lemma}
	\label{lem-union}
	For any languages $K$ and $L$ $\blim{K\cup L} = \blim{K}\cup\blim{L}$.
\end{lemma}
\begin{proof}
	Assume $w$ is a limit of $L$. Then for each $w_0<_pw$, there exists some $u\in L$ with $w_0<_p u$.
	Since then $u\in K\cup L$ as well, $w$ is a limit of $K\cup L$ as well.
	
	For the other direction, assume $w$ is a limit of $K\cup L$. Then for each $w_0<_pw$, there exists
	some $u\in K\cup L$ with $w_0<_pu$. Thus, either there exists infinitely many prefixes $w_0$ of $w$
	for which there exists some $u\in K$ with $w_0<_pu$ or there exists infinitely many prefixes $w_0$ of 
	$w$ for which there exists some $u\in L$ with $w_0<_pu$. In the former case, $w$ is a limit of $K$,
	in the latter, $w$ is a limit of $L$.
\end{proof}

\begin{lemma}
	\label{lem-limits-of-product}
	$\blim {KL} = \blim{K} \cup K\blim{L}$ if $K, L\neq \emptyset$.
\end{lemma}
\begin{proof}
	 Let $u\in K$ be a word and $w$ be a limit of $L$. To prove that $uw$ is a limit of $KL$ we only have to show that for each prefix $w'$ of $uw$ there exist a word $w^*\in KL$ with $w' <_p w^*$. Let $w'\in \mathbf{Pref}(uw)$, and since it is enough to see the prefixes which are longer than $u$, the word $w'$ can be written as $uw_0$. Since $w$ is a limit of $L$ there exists a word $v\in L$ such that $w_0<_p v$. Thus there is a word $uv\in KL$ such that $w'<_puv$.
	 
	 Now let $w$ be a limit of $K$. To prove that $w$ is a limit of $KL$, let $u$ be a word in $L$ and $w_0$ be a prefix of $w$. Since $w$ is a limit of $K$ there exists a word $v\in K$ with $w_0 <_p v$ by definition. So $vu \in KL$ and $w_0$ is a prefix of $vu$ as well.
	 
	 For the other direction, we have to prove that there are no more limits of $KL$. Let $w$ be a limit of $KL$ and $w_i$ be the prefix of $w$ with length $i$. Since $w$ is a limit of $KL$ there exist a word $u_iv_i\in KL$ such that $w_i <_p u_iv_i$ where $u_i\in K$ and $v_i\in L$. Consider for each $i>0$ the lengths of these words $u_i$.
	 There are two cases: either there is a finite upper bound on $|u_i|$ or there is not.
	 
	 If $|u_i|$ is not upperbounded, then $w$ is a limit of $K$, since for each prefix $w_i$ there exists
	 some long enough $u_j$ with $w_i<_pu_j$.

	In the case where the lengths of these $u_i$ words is bounded, let $\ell = \max|u_i|$ be the maximal length.
	Since there are only finitely many words of length at most $\ell$, there has to be some $u=u_j$ such that
	$u=u_i$ for infinitely many indices $i$. Hence in particular, $w=uw'$ for some $w'\in\Sigma^\omega$.
	We show that $w'\in\blim L$, yielding $w\in K\blim L$. Indeed, if $w^*<_pw'$ is a prefix of $w'$, then
	$uw^*<_pw$ and thus there exists some $v_i$ with $uw^*<_puv_i$, that is, $w^*<_pv_i$ and so $w^*$ is a limit of $L$.
\end{proof}

\begin{corollary}
	\label{cor-la-al}
	For any language $L$ and letter $a\in\Sigma$,
	$\blim{L}=\blim{La}$ and $a\cdot\blim{L}=\blim{aL}$.
\end{corollary}

\begin{corollary}
	\label{cor-uLv}
	For any language $L\subseteq\Sigma^*$ and words $u,v\in\Sigma^*$, $\blim{uLv}=u\cdot \blim{L}$.
\end{corollary}

\subsection{Unique limits}
In this part we establish the decidability of the problem whether a context-free language has a unique limit.
(In this case, the limit itself is computable as well, thanks to Lemma~\ref{lem-sequence}.)

In the rest of the paper when grammars are involved, we assume the grammar $G=(N,\Sigma,P,S)$
contains no left recursive nonterminals,
and for each $X\in N$, $X$ is usable and $L(X)$ is an infinite language of nonempty words.
Moreover, each nonterminal but possibly $S$ is assumed to be recursive.
Any context-free grammar can effectively be transformed into such a form, see e.g.~\cite{GelleIvanTCS}.

It is also known~\cite{10.1007/978-3-642-22321-1_19} that if the context-free grammar $G$ generates
a scattered language, then for each recursive nonterminal $X$ there exists a unique (and computable)
primitive word $u_X$ such that whenever $X\Rightarrow^+uX\alpha$ for some $u\in\Sigma^*$ and sentential form
$\alpha$, then $u\in u_X^+$. Moreover, for each pair $X\approx Y$ of recursive nonterminals there exists
a (computable) word $u_{X,Y}\in\Sigma^*$ such that whenever $X\Rightarrow^+ uYv$, then $u\in u_{X,Y}u_Y^*$.

\begin{lemma}
	\label{lem-finite-strict-unique-limit}
	The word $w\in\Sigma^\omega$ is the unique limit of an infinite language $L$ if and only if for each $w_0<_p w$ there exists only finitely many words $u\in L$ such that $u<_s w_0$ or $w_0<_s u$.
\end{lemma}
\begin{proof}
	We will see just the case where $w_0 <_s u$, the other one can be done analogously.
	
	Suppose for the sake of contradiction there exist infinitely many words $u\in L$ with $w_0 <_s u$. 
	Let $w_0$ be the shortest such word, it can be written as $w_0 = w_0'a$. Since there are just finitely many words $u\in L$ with $w_0'<_s u$, it has to be the case that $w_0'<_p u$ and $w_0'a<_s u$ for infinitely many $u\in L$. Then there exists a letter $b\in\Sigma$ such that $b>a$ and infinitely many words $u\in L$ such that $w_0'b <_p u$. But any limit of these words is in $w_0'b\Sigma^\omega$
	(and by Lemma~\ref{lem-infinite} at least one limit exists), which cannot be equal to $w$, so the language $L$ has two different limits which is a contradiction. 
\end{proof}

\begin{lemma}
	\label{lem-unique-limit-embeds-into-omega-plus-minusomega}
	If $L$ has a unique limit $w$, then $o(L_{<w}) \preceq \omega$ and $o(L_{>w}) \preceq -\omega$.
	Moreover, in this case both $o(L_{<w})$ and $o(L_{>w})$ are effectively computable
	(and hence so is $o(L)=o(L_{<w})+o(L_{>w})$).
\end{lemma}
\begin{proof}
	In the $o(L_{<w})$ case, if $L_{<w}$ is finite (and thus its size is computable),
	we are done. Otherwise $L_{<w}$ is infinite, which means it has a limit, and this limit has to be the $w$ (since it is unique for $L$). By Lemma~\ref{lem-sup-is-limit} the supremum
	$\bigvee L_{<w}$ of this language is $w$ as well, moreover,
	each infinite $K\subseteq L_{<w}$ has $\bigvee K = w$, so the order type of $L_{<w}$ has to be $\omega$.
	
	In the $o(L_{>w})$ case, if $L_{>w}$ is finite, it can be embedded into $-\omega$ so we are done. Otherwise $L_{>w}$ is infinite, so it has to have a limit which is $w$. Since this limit should be an infimum of a descending chain, each infinite $K\subseteq L_{>w}$ has $\bigwedge K = w$, so the order type of $L_{>w}$ has to be $-\omega$.
\end{proof}

Before proceeding to the case of concatenation, we recall the notion of prefix chains from~\cite{GelleIvanGandalf}.
For a word $w\in\Sigma^\omega$, let $\mathbf{Pref}(w)$ stand for the set $\{u\in\Sigma^*:u<_pw\}$ of
proper prefixes of $w$. A language $L\subseteq\Sigma^*$ is called a \emph{prefix chain} if $L\subseteq\mathbf{Pref}(w)$
for some $\omega$-word $\omega$. Lemma 2 from~\cite{GelleIvanGandalf} states that it is decidable for
any context-free language $L$ whether $L$ is a prefix chain and if so, a suitable $w\in\Sigma^\omega$ 
can be effectively computed.

\begin{lemma}
	\label{lem-x1x2}
	Assume we know that for the nonterminals $X_1$ and $X_2$
	whether the languages $L_1=L(X_1)$ and $L_2=L(X_2)$ have a unique limit.	
	Then it is computable whether $L=L(X_1X_2)$ has a unique limit.
\end{lemma}
\begin{proof}
	By assumption, the nonterminals $X_1$ and $X_2$ each generate an infinite language so they have at
	least one limit by Lemma~\ref{lem-infinite}.	
	By Lemma~\ref{lem-limits-of-product}, if either $L_1$ or $L_2$ has at least two limits, then so have
	$L$ and we can stop.
	
	Assume both $L_1$ and $L_2$ have a unique limit.
	As they are both context-free languages,
	their limits are computable regular words by Lemma~\ref{lem-sequence}. Let $u_1v_1^\omega$ and $u_2v_2^\omega$
	respectively be the limits of $L_1$ and $L_2$.
	If $L_1$ is not a prefix chain, that is, $x<_sy$ for some $x,y\in L_1$, then both $xu_2v_2^\omega$
	and $yu_2v_2^\omega$ are limits of $L$ by Lemma~\ref{lem-limits-of-product}
	and these two words are different by $x<_sy$. So in this case we can stop.
	
	From this point we can assume that $L_1$ is a prefix chain, that is, $L_1\subseteq\mathbf{Pref}(u_1v_1^\omega)$.
	By Lemma~\ref{lem-limits-of-product}, $u_1v_1^\omega$ is also a limit of $L$,
	and for each $u\in L_1$, the word $uu_2v_2^\omega$ is also a limit of $L$.
	Thus, we have to decide whether $u_1v_1^\omega=uu_2v_2^\omega$ holds for each $u\in L_1$.
	
	Now consider the direct product automaton $M=M_{u_1,v_1}\times M_{u_2,v_2}$, where in the automaton
	corresponding to $L_2$ we use primed states $q'$ in place of each state $q$.
	
	Obviously, $u_1v_1^\omega\neq uu_2v_2^\omega$ if and only if from the state $(q_0\cdot u,q'_0)$
	some state of the form $(\bot,q')$ or $(q,\bot')$ is reachable for some $q\neq\bot$.
	Thus, it suffices to determine the set of states $q$ in $M_1$ for which $q_0\cdot u=q$ for some $u\in L_1$,
	that is, for which $M_1(q)\cap L_1$ is nonempty, which is decidable since $L_1$ is context-free
	and $M_1(q)$ is regular. Then, for each such state $q$ we test whether a state of the form
	$(\bot,p')$ or $(p,\bot')$ is reachable from $(q,q'_0)$ and if so, then $uu_2v_2^\omega\neq u_1v_1^\omega$
	for some $u\in L_1$ and thus $L$ has at least two different limits.
	
	Now assume $uu_2v_2^\omega=u_1v_1^\omega$ for each $u\in L_1$. We claim that in this case $L$ has the unique
	limit $u_1v_1^\omega$. To see this, we first prove that for each prefix $w_0<_pu_1v_1^\omega$,
	there are only a finite number of words $v\in L$ with either $w_0<_sv$ or $v<_sw_0$.
	Suppose for the sake of contradiction that for some prefix $w_0$ of $u_1v_1^\omega$ there are
	infinitely many such words and let $w_0$ be the shortest such prefix. 
	By assumption, there are infinitely many words $xy\in L$, $x\in L_1$, $y\in L_2$,
	with either $xy<_sw_0$ or $w_0<_sxy$. Since $L_1\subseteq\mathbf{Pref}(u_1v_1^\omega)$ and $w_0<_pu_1v_1^\omega$,
	this can happen only if $x<_pw_0$. Since there are only finitely many prefixes of $w_0$,
	for some $u<_pw_0$ in $L_1$ there has to be an infinite number of words $y\in L_2$ with
	either $uy<_sw_0$ or $w_0<_suy$. Let us write $w_0=uv$.
	The condition $uy<_sw_0$ or $w_0<_suy$ is then equivalent $y<_sv$ or $v<_sy$ for infinitely many $y\in L_2$.
	But since we know that $uu_2v_2^\omega=u_1v_1^\omega$, $v$ is a prefix of $u_2v_2^\omega$,
	which is the unique limit of $L_2$ and thus only finitely many words $y\in L_2$ exist with 
	$y<_sv$ or $v<_sy$ by Lemma~\ref{lem-finite-strict-unique-limit}, yielding a contradiction.
	
	Hence in this case, $u_1v_1^\omega$ is the unique limit of $L$.	
\end{proof}

\begin{corollary}
	\label{cor-unique-alpha-limit}
	Assume $n\geq 0$ and $X_1,\ldots,X_n\in N\cup\Sigma$ are symbols so that for each $X_i$ we know
	whether $L(X_i)$ has a unique limit. Then it is computable whether $L(X_1\ldots X_n)$ has a unique limit.
\end{corollary}
\begin{proof}
	Let us introduce the fresh nonterminals $Y_1,\ldots,Y_{n-1}$ and productions
	$Y_1\to X_1Y_2$, $Y_2\to X_2Y_3$,\ldots, $Y_{n-1}\to X_{n-1}X_n$. Applying Lemma~\ref{lem-x1x2}
	or Corollary~\ref{cor-la-al} (depending on whether $X_i$ is a nonterminal or a letter)
	for the nonterminals $Y_{n-1}$, $Y_{n-2}$, \ldots, $Y_1$ in this order we can successively decide whether
	each $L(Y_i)$ has a unique limit, proving the statement since $L(X_1\ldots X_n)=L(Y_1)$.
\end{proof}

\begin{corollary}
	\label{cor-alpha}
	Let $X$ be a nonterminal $X\to\alpha_1~|~\ldots~|~\alpha_k$ be the collection of all the escaping productions
	with left-hand side $X$. Assume we already know for each $Y\prec X$ whether $L(Y)$ has a unique limit.
	Then it is computable whether $L(\{\alpha_1,\ldots,\alpha_k\})$ has a unique limit.
\end{corollary}
\begin{proof}
	By Corollary~\ref{cor-unique-alpha-limit}, it is computable for each $\alpha_i$ whether each $L(\alpha_i)$
	is finite or has a unique limit.
	If not, then neither has their union (by Lemma~\ref{lem-union}).
	If each of the languages $L(\alpha_i)$ is either finite or has a unique limit, then their union
	has a unique limit if and only if all the limits are the same.
	But this is decidable since these languages are context-free, hence their unique limit is a computable
	regular word by Lemma~\ref{lem-sequence}, and the equivalence of these words is decidable.
\end{proof}

\begin{lemma}
	\label{lem-recursive-nonprefix-alpha-x-beta}
	Assume $X$ is a recursive nonterminal, $L(X)$ is not a prefix chain
	and for some nonterminal $X'\approx X$ there exists a 
	production $X'\to \alpha X''\beta$ with $\beta$ containing at least one nonterminal.
	
	Then $L(X)$ has at least two limits.	
\end{lemma}
\begin{proof}
	Let $u<_sv$ be members of $L(X)$. Since $\beta$ contains a nonterminal, $L(\beta)$ is
	infinite and has a limit $w$ by Lemma~\ref{lem-infinite}. By the conditions on the
	recursive nonterminal $X$, we get $X\Rightarrow^*u_1Xu_2\beta u_3$ for some words $u_1,u_2,u_3\in\Sigma^*$.
	By Lemma~\ref{cor-uLv}, both $u_1uu_2w$ and $u_1vu_2w$ are limits of $L(X)$ and they are
	distinct by $u<_sv$.
\end{proof}

\begin{lemma}
	\label{lem-recursive-x}
	Assume $L(G)$ is scattered.
	Then it is decidable for each nonterminal $X$ whether $L(X)$ has a unique limit.
\end{lemma}
\begin{proof}
	We prove the statement by induction on $\prec$. So let $X$ be a nonterminal and 
	assume we already know for each $Y\prec X$ whether $L(Y)$ has a unique limit.
	
	If $X$ is nonrecursive, and $X\to\alpha_1~|~\ldots~|~\alpha_k$ are all the alternatives of $X$,
	then the question is decidable by Corollary~\ref{cor-alpha}.
	
	So let $X$ be a recursive nonterminal. If $L(Y)$ has at least two limits for some $Y\prec X$,
	then by Corollary~\ref{cor-uLv} so does $L(X)$ and we are done. So we can assume from now on
	that each $L(Y)$ with $Y\prec X$ has exactly one limit. This limit is a computable regular word.
	
	If $L(X)$ is a prefix chain, then its supremum is its
	unique limit and we can stop. So we can assume that $L(X)$ is not a prefix chain.
	Now if there exist some production of the form $X'\to \alpha X''\beta$ with $X'\approx X''\approx X$
	and $\beta$ containing at least one nonterminal, then by Lemma~\ref{lem-recursive-nonprefix-alpha-x-beta},
	$L(X)$ has at least two limits and we can stop. 
	
	Otherwise, we can assume that each component production in the component of $X$ has the form
	$X'\to \alpha X''u$ for some $u\in\Sigma^*$, $X'\approx X''\approx X$. Since $L(X)$ is scattered,
	for each such $\alpha$ it has to be the case that $L(\alpha)\subseteq u_{X',X''}^{}u_{X''}^*$.
	
	Now let $X'\approx X$ be a nonterminal and $\alpha_1,\ldots,\alpha_k$ all the escaping alternatives
	of $X'$. If $L(\{\alpha_1,\ldots,\alpha_k\})$ has at least two limits, then so does $L(X')$ and $L(X)$
	and we are done. Otherwise, if $L(\{\alpha_1,\ldots,\alpha_k\})$ is infinite, then its unique limit
	is a computable word. On the other hand, for each recursive nonterminal $X'$ the word $u_{X'}^\omega$
	is a limit of $L(X')$, and by $L(\{\alpha_1,\ldots,\alpha_k\})\subseteq L(X')$, the two limits has to
	coincide (which is decidable). If they are not the same, then again, $L(X)$ has at least two limits
	and we can stop.
	
	Hence we can assume that for each $X'\approx X$, the language $L(X')$ has the limit $u_{X'}^\omega$
	which is the same as the unique limit of $L(\{\alpha_1,\ldots,\alpha_k\})$ if this latter language is infinite.	
	
	We claim that in this case, $L(X)$ has the unique limit $u_X^\omega$. To see this, we apply
	Lemma~\ref{lem-finite-strict-unique-limit} and show that for each prefix $w_0$ of $u_X^\omega$,
	there are only finitely many words $u\in L(X)$ with either $w_0<_su$ or $u<_sw_0$.
	
	Assume to the contrary that $w_0<_pu_X^\omega$ and there are infinitely many words $u\in L(X)$ with
	either $w_0<_su$ or $u<_sw_0$. Each word $u\in L(X)$ can be derived from $X$ using a leftmost
	derivation sequence resulting in a sentential form $u_X^tu_{X,X'}\alpha v$ for some $t\geq 0$
	so that $X'\to\alpha$ is an
	escaping production from the component of $X$ and $u\in u_X^tu_{X,X'}L(\alpha)v$.
	Since $u$ and $w_0<_pu_X^\omega$ are not related by $<_p$, we have an upper bound for $t$,
	which, as $G$ does not contain left-recursive nonterminals, places an upper bound for $|v|$.
	Hence, there are only finitely many possibilities for picking $t\geq 0$, $X'$, $\alpha$ and $v$,
	thus for some combination of them, there are infinitely many such words $u$ belonging to the
	same language $u_X^tu_{X,X'}L(\alpha)v$. So we can write each such $u$ as $u=u_X^tu_{X,X'}u'v$
	with $u'\in L(\alpha)$, and we can write $w_0$ as $w_0=u_X^tu_{X,X'}w'_0$, that is, $w'_0<_p u_{X'}^\omega$.
	This yields that $u'v<_sw'_0$ or $w'_0<_su'v$ for infinitely many words $u'\in L(\alpha)$.
	Thus, there are infinitely words $u'\in L(\alpha)$ of length at least $|w'_0|$ with either $u'v<_sw'_0$ or
	$w'_0<_su'v$, hence with either $u'<_sw'_0$ or $w'_0<_su'$, which is a contradiction, since
	by Lemma~\ref{lem-finite-strict-unique-limit} this would yield that $L(\alpha)$ has at least two
	limits, which we already handled in a former case.
\end{proof}

\subsection{Finitely many limits}
In this part we extend the results of the previous subsection for the context-free languages having
a finite number of limits and get the main result of the paper: it is decidable whether a context-free
language has a scattered order type of rank at most one, and if so, then its order type is 
effectively computable.

\begin{lemma}
	\label{lem-uvomega-islimit-decidable}
	It is decidable for any context-free language $L$ and regular word $w=uv^\omega$
	whether $w$ is a limit of $L$.
\end{lemma}
\begin{proof}
	Let $L\subseteq\Sigma^*$ be a context-free language and
	consider the generalized sequential mapping $f:\Sigma^*\to a^*$ defined as
	\[f(x)=\begin{cases}a\cdot f(y)&\hbox{if }x=vy\hbox{ for some }y\in\Sigma^*\\\varepsilon&\hbox{otherwise}.\end{cases}\]
	Now for any word $x$, $f(x)=a^n$ for the unique $n$ such that $x=v^ny$ for some $y$ not having $v$ as prefix.
	Thus, $v^\omega$ is a limit of a language $L'$ if and only if $f(L')$ is infinite;
	hence, by Corollary~\ref{cor-uLv}, $w=uv^\omega$ is a limit of $L$ if and only if $f(u^{-1}L)$ is infinite.
	Since the class of context-free languages is effectively closed under left quotients and generalized
	sequential mappings, and their finiteness problem is decidable, the claim is proved.
\end{proof}

\begin{lemma}
	\label{lem-recursive-one-or-infinite}
	If $X$ is a recursive nonterminal, then $u_X^\omega \in \blim{X}$ and if some $w\neq u_X^\omega$ is also a member of $\blim{X}$, then $\blim{X}$ is infinite.
\end{lemma}
\begin{proof}
	Since $X\Rightarrow^* u_X^nX\alpha$ holds for each recursive nonterminals, there exists a word $v\in L(X)$ for each $u\in\mathbf{Pref}(u_X^\omega)$ such that $u<_p v$.
	
	If $w\neq u_X^\omega$ is also a limit of $L(X)$, then it can be written as $w=ubw'$, where $u<_pu_X^\omega$, $ub\not<_pu_X^\omega$ and $w'\in\Sigma^\omega$. So if we consider a derivation of the form
	$X\Rightarrow^* u_X^nX\alpha$, we get that each $(u_x^n)^kubw'$, $k>0$ is a limit, thus $\blim{X}$ is infinite
	as these words are pairwise different.
\end{proof}

\begin{lemma}
	\label{lem-kvomega-is-finite}
	Assume $K$ is a context-free language and $v$ is a nonempty word. Then it is decidable whether $Kv^\omega$ is finite
	and if so, its members (which are regular words) can be effectively enumerated.
\end{lemma}
\begin{proof}
Consider the generalized sequential mapping $f: \Sigma^*\to \Sigma^*$ defined as
$$f(x) = \begin{cases}
f(y)&\hbox{if } x=v^Ry \hbox{ for some } y\in\Sigma^*\\
x&\hbox{otherwise.}
\end{cases}$$

Now for any word $x$, $f(x)=y$ for some $y$ not having $v^R$ as prefix such that $x=(v^R)^ky$ for some $k\geq 0$,
that is, $f$ strips away the leading $v^R$s of its input.
So, we have that $\bigl(f(K^R)\bigr)^R$ consists of those words we get from members of $K$, stripping
away their trailing $v$s.
Now $u\in \bigl(f(K^R)\bigr)^R$ if and only if $u$ does not end with $v$ and $uv^\omega\in Kv^\omega$.
Moreover, $Kv^\omega$ is finite if and only if there exist some $u_1, u_2, \ldots, u_n\in \Sigma^*$ such that for each $i\in[n]$ the word $u_i$ does not end with $v$ and  $Kv^\omega = \{u_iv^\omega ~|~ i\in [n]\}$.
So we get that $Kv^\omega$ is finite if and only if so is $\bigl(f(K^R)\bigr)^R$ which is decidable since the class of context-free
languages is effectively closed under reversal and generalized sequential mappings, and their finiteness problem is also
decidable. In this case, members of $\bigl(f(K^R)\bigr)^R=\{u_1,\ldots,u_n\}$ can also be effectively enumerated and
$Kv^\omega=\{u_jv^\omega:j\in[n]\}$.
\end{proof}

\begin{lemma}
	\label{lem-kl-finite}
	Assume $K$ and $L$ are context-free languages such that $\blim{K} = \{u_iv_i^\omega ~|~ 1\leq i\leq k\}$ and $\blim{L} = \{u'_jw_j^\omega ~|~ 1\leq j\leq \ell \}$ are finite sets of regular words. Then it is decidable whether $\blim{KL}$ is finite and if so, then it is a computable (finite) set of regular words. 
\end{lemma}
\begin{proof}
	By Lemma~\ref{lem-limits-of-product} $\blim{KL}=\blim{K}\cup K\blim{L}$. Since $\blim{K}$ is a finite set
	(and is of course computable since it is given as input),
	we only have to deal with $K\cdot\blim{L}$.
	Since \[K\cdot\blim{L}=K\cdot\bigl\{u_j'w_j^\omega:j\in[\ell]\bigr\}=\mathop\bigcup\limits_{j\in[\ell]}(Ku_j')w_j^\omega\]
	and this union is finite if and only if so is each language $(Ku'_j)w_j^\omega$, which is decidable by
	Lemma~\ref{lem-kvomega-is-finite} (since the languages $Ku'_j$ are each context-free), we get decidability
	and even computability if each of them is finite.		
\end{proof}

\begin{corollary}
	\label{cor-finite-alpha-limit}
	Assume $n\geq 0$ and $X_1,\ldots,X_n\in N\cup\Sigma$ are symbols so that for each $X_i$,
	$\blim{X_i}$ is a known finite set.
	Then it is decidable whether $L(X_1\ldots X_n)$ has a finite number of limits and if so,
	$\blim{X_1\ldots X_n}$ is effectively computable.
\end{corollary}
\begin{proof}
	Let us introduce the fresh nonterminals $Y_1,\ldots,Y_{n-1}$ and productions
	$Y_1\to X_1Y_2$, $Y_2\to X_2Y_3$,\ldots, $Y_{n-1}\to X_{n-1}X_n$. Applying Lemma~\ref{lem-kl-finite}
	or Corollary~\ref{cor-la-al} (depending on whether $X_i$ is a nonterminal or a letter)
	for the nonterminals $Y_{n-1}$, $Y_{n-2}$, \ldots, $Y_1$ in this order we can decide whether
	each $L(Y_i)$ has a finite number of limits, and if so, we compute $\blim{Y_i}$ as well,
	proving the statement since $L(X_1\ldots X_n)=L(Y_1)$.
\end{proof}

\begin{corollary}
	\label{lem-limx-finite}
	It is decidable for any nonterminal $X$ whether $L(X)$ has a finite number of limits.
\end{corollary}
\begin{proof}
	There are two cases: either $X$ is recursive or not. If $X$ is recursive, then by Lemma~\ref{lem-recursive-one-or-infinite}
	$L(X)$ has a finite number of limits if and only it has a unique limit which is decidable due to Lemma~\ref{lem-recursive-x}.
	Hence we can decide for each nonterminal $X\neq S$ whether $\blim{ L(X) }$ is finite.
	
	Now suppose $X$ is not recursive (thus $X=S$ as $G$ is in normal form)
	and let $\balpha_1,\ldots,\balpha_t$ be all the alternatives of $S$.	
	By Corollary~\ref{cor-finite-alpha-limit} for each $\alpha_i$ $(1\leq i\leq t)$ it is decidable whether $L(\alpha_i)$ has a finite number of limits. If one of them has infinite limits than so has $X$.
	So we can assume that each $L(\alpha_i)$ has a finite number of limits and we can
	even compute each $\blim{\alpha_i}$.
	By Lemma~\ref{lem-union}, $\blim{X} = \bigcup_{i\in [t]}\blim{\alpha_i}$.
\end{proof}

\begin{theorem}
	\label{thm-finite-limits-imply-computability}
	Suppose $L$ is a context-free language having a finite number of limits.
	Then $o(L)$ is effectively computable and is scattered of rank at most one.
\end{theorem}
\begin{proof}
	We prove the statement by induction on the number of limits.
	
	If $L$ has no limits, then it is finite by Lemma~\ref{lem-infinite}, and so $o(L)=|L|$ is computable.
	
	If $L$ has a unique limit, then $o(L)$ can be embedded into $\omega+-\omega$ and is computable
	by Lemma~\ref{lem-unique-limit-embeds-into-omega-plus-minusomega}. Moreover, it is decidable
	whether $L$ has a unique limit by Lemma~\ref{lem-recursive-x}.
	
	Now assume $L$ has at least two limits. Since $L$ is infinite, we can compute a regular limit
	of the form $w=uv^\omega$ for $L$ by Lemma~\ref{lem-sequence}.
	By Lemma~\ref{lem-uvomega-islimit-decidable}, it is decidable whether $w$ is a limit
	of either $L_{<w}$ or $L_{>w}$ or both of them. (By Lemma~\ref{lem-union}, $w$ is a limit of 
	at least one of them.) If $w$ is not a limit of $L_{<w}$ ($L_{>w}$, resp.), then this language
	has a smaller number of limits than $L$ and we can proceed by induction. Suppose now $w$ is a limit of
	$L_{<w}$ -- it has to be $w=\bigvee L_{<w}$ then.
	If $L$ has a limit which is larger than $w$
	(that is, $L_{>w}$ is infinite and either $w$ is not a limit of $L_{>w}$ or $L_{>w}$ has at least two
	limits -- this is decidable as well),
	then $L_{<w}$ has a smaller number of limits than $L$
	(since no limit of $L_{<w}$ can be strictly larger than its supremum)
	and we can proceed again by induction and get that $L_{<w}$ is computable.
	It is also decidable whether $L_{<w}$ has only one limit and if so, its order type is also computable
	and we are done.
	
	The last case is when $w=\bigvee L_{<w}$ is the largest limit of $L$ and $L_{<w}$ has at least two limits.
	Thus, there exists some limit $w'$ of $L_{<w}$ and an integer $n\geq 0$ such that $w'<_suv^n$,
	or equivalently, $L_{<uv^n}$ is infinite for some $n\geq 0$. We can compute (say, the least) such $n$
	by starting from $n=0$ and iterating, eventually we will find an integer $n$ with this property.
	Then, $L_{<uv^n}$ has a smaller number of limits than $L_{<w}$ so we can use induction and compute $o(L_{<uv^n})$;
	also, ${L_{<w}}_{\geq uv^n}$ also has a smaller number of limits than $L_{<w}$ (since $w'$ is missing) and we can
	apply induction to this half as well and compute its order type. Then, $o(L_{<w})$ is the sum of
	the two already computed order types.
	
	Repeating the same argument (by appropriate modifications: taking infimum instead of supremum,
	splitting the case when $w$ is the least limit of $L$) we get that $o(L_{>w})$ is also computable,
	and $o(L)$, being the sum $o(L_{<w})+o(L_{>w})$, is hence computable as well.
	
	We also got that the order type of such a language has to be a finite sum of
	the order types $\omega$, $-\omega$ and $1$, that is, has to have rank at most $1$.
\end{proof}

\begin{corollary}
	\label{cor-hausdorff-one-computable}
	Suppose $L$ is a scattered context-free language of rank at most one.	
	Then $o(L)$ is effectively computable.
\end{corollary}
\begin{proof}
	If $o(L)\in\{\omega,-\omega\}$, then $L$ has one limit, while if $o(L)$ is finite, then it has no limits.
	Since scattered order types of rank at most one are finite sums of the order types
	$\omega$, $-\omega$ and $1$, thus scattered languages of rank at most one are finite unions of
	languages of order type $\omega$, $-\omega$ or $1$, by Lemma~\ref{lem-union} we get that
	such languages have a finite number of limits, and thus their order type is effectively computable
	by Theorem~\ref{thm-finite-limits-imply-computability}.
\end{proof}

\begin{corollary}
	\label{cor-main}
	For any context-free language $L$, it is decidable whether $L$ is a scattered language of rank at most one,
	and if so, $o(L)$ can be effectively computed (as a finite sum of the order types $1$, $\omega$ and $-\omega$).
\end{corollary}

\section{Conclusion}
We showed that it is decidable whether a context-free ordering is scattered of rank at most one, and
if so, then its order type is effectively computable as a finite sum of the order types $1$, $\omega$
and $-\omega$. This complements our earlier result~\cite{GelleIvanGandalf} that for scattered
context-free orderings of rank (at least) two, it is undecidable whether their order type
is $\omega+(\omega+\zeta)\times\omega$, thus the order type is not computable, even if the grammar
in question is a so-called prefix grammar.

An interesting question for further study is whether the rank of a scattered context-free ordering
is computable. Another, maybe easier one is to determine which rank-two scattered orderings are
context-free (as there are uncountably many such orderings, the vast majority of them cannot be
context-free).

A relatied notion is that of tree automatic orderings: these are the order types of regular tree languages
equipped with the lexicographic ordering (on trees). Through derivation trees, there is a tight connection
between context-free string languages and regular tree languages but as the two orderings differ (lexicographic
ordering of trees vs their frontiers), it is unclear whether there is a nontrivial inclusion between these
two classes of orderings (or at least for the scattered case).
\bibliography{biblio}{}
\bibliographystyle{splncs04}
\end{document}